\documentclass[conference]{IEEEtran}
\usepackage{amssymb,amsthm,comment}
\usepackage{amsmath,amssymb}
\usepackage{bm,amsfonts}
\usepackage{nidanfloat}
\usepackage{enumerate}
\usepackage{graphicx}
\usepackage[ruled,vlined,linesnumbered]{algorithm2e}

\DeclareMathOperator*{\argmax}{\rm argmax}
\DeclareMathOperator*{\argmin}{\rm argmin}

\newcommand{\E}{\mathrm{E}}

\newcommand{\n}{\nonumber}
\newcommand{\nn}{\nonumber\\}

\newcommand{\rd}{\mathrm{d}}

\newcommand{\lo}{\mathrm{O}}

\newcommand{\inner}[1]{\langle #1\rangle}

\newcommand{\cX}{\mathcal{X}}

\newcommand{\com}{\,,}
\newcommand{\per}{\,.}

\newtheorem{lemma}{Lemma}
\newtheorem{theorem}{Theorem}

\newtheorem{remark}{Remark}

\newenvironment{proof2}[1]{{\it Proof #1:}}{\vspace{1.5mm}}
\newcommand{\fqed}{\hfill $\blacksquare$ \par}
\def\dqed{\relax\tag*{\fqed}}

\newcommand{\tx}{\tilde{x}}

\newcommand{\fl}[1]{\lfloor#1\rfloor}
\newcommand{\uI}{\underline{I}}
\newcommand{\oI}{\overline{I}}
\renewcommand{\baselinestretch}{0.99}

\begin{document}

\title{
Variable-to-Fixed Length Homophonic Coding Suitable for
Asymmetric Channel Coding
}

\author{\IEEEauthorblockN{Junya Honda \qquad Hirosuke Yamamoto}
\IEEEauthorblockA{Graduate School of Frontier Sciences,\\
The University of Tokyo,\\
Kashiwa-shi Chiba 277--8561, Japan\\
Email: honda@it.k.u-tokyo.ac.jp, Hirosuke@ieee.org}
}

\maketitle

\begin{abstract}
In communication through asymmetric channels
the capacity-achieving input distribution is not uniform in general.
Homophonic coding is a framework to invertibly convert a (usually uniform) message into a sequence with
some target distribution, and is
a promising candidate to generate codewords with the nonuniform target distribution for asymmetric channels.
In particular, a Variable-to-Fixed length (VF) homophonic code can be
used as a suitable component for channel codes
to avoid decoding error propagation.
However, the existing VF homophonic code requires the knowledge of
the maximum relative gap  of probabilities between
two adjacent sequences beforehand, which is an unrealistic assumption for long block codes.
In this paper\footnote{%
This is the full version of the paper
to appear in IEEE International Symposium on Information Theory (ISIT2017) with some minor corrections.} we propose a new VF homophonic code without such a requirement by allowing
one-symbol decoding delay.
We evaluate this code theoretically and experimentally
to verify its asymptotic optimality.
\end{abstract}

\IEEEpeerreviewmaketitle

\allowdisplaybreaks[4]

\section{Introduction}
In communication through asymmetric channels
the capacity-achieving input distribution
is not uniform in general.
On the other hand, in most practical codes such as linear codes
all symbols appear with almost the same frequency
and some modification is necessary to use them as optimal codebooks.
Although it is 
known that
biased codewords can be generated
from an auxiliary code over an extended alphabet
based on Gallager's nonlinear mapping \cite[p.\,208]{gallager_map},
its complexity becomes very large
when the target distribution is not expressed
in simple rational numbers.


A promising solution to this problem is
to generate such biased codewords by homophonic coding.
Homophonic coding is a framework to
invertibly convert a message with distribution $P_U$
into another sequence with target distribution $P_X$.
This framework is intuitively similar to a dual of lossless compression,
where a biased source sequence is encoded into an almost uniform sequence.
In fact,
the inverse of a lossless code based on an LDPC matrix
is used to construct capacity-achieving channel
code in
\cite{miyake_channel_general}
although a practical algorithm for this code has not been known.

A homophonic code is called {\it perfect}
if the generated sequence exactly follows the target distribution.
A perfect Fixed-to-Variable length (FV) homophonic code%
\footnote{
A similar FV code is proposed in \cite{boc_fv} later
that is asymptotically perfect.}
can be constructed \cite{homophonic_interval} based
on an interval algorithm similar to a random number generator \cite{interval}.
This code can actually be applied to generation of biased codewords
for LDPC codes and polar codes
\cite{honda_phd}\cite{wang_memory}
to achieve the capacity.
These codes require a homophonic code
applicable to a non-i.i.d.~sequence
to generate a codeword with some structure (such as parity-check constraints).
On the other hand,
a general channel coding framework is proposed in \cite{boc_chain}\cite{howto_capacity}
where any homophonic code for i.i.d.~sequences can be used
at a cost that each codeword consists of some blocks
and the entire code length becomes large.

Although the FV homophonic code in \cite{homophonic_interval}
is perfect and
achieves the asymptotic bound on the coding rate,
it is not appropriate to use this code
as a component of a channel code.
This is because
an FV channel code is hard
to run in parallel and
suffers decoding error propagation,
which is a serious problem for channel coding.
Therefore
it is desirable to use a Fixed-to-Fixed length (FF) or Variable-to-Fixed length (VF) homophonic code
to avoid such error propagation.

Since it is difficult to construct
a perfect homophonic code in the VF and FF frameworks\footnote{Although
a perfect FF homophonic code
is also constructed in \cite{homophonic_interval},
it suffers decoding errors with asymptotically vanishing but positive probability.},
there are some studies on homophonic codes
whose output distribution asymptotically matches the target distribution.
Such an asymptotically matching distribution is practically sufficient
for the channel coding application,
in contrast to the
original motivation for homophonic coding \cite{gunther}\cite{homophonic}
where an exactly matching random sequence is required to apply to cryptography.

B\"ocherer and Mathar
considered homophonic coding with the name {\it distribution matching}
and proposed
a VF homophonic code
based on Huffman coding with an approximation of the target distribution by a dyadic distribution
\cite{homophonic_huffman}\cite{bocherer_phd}.
This code requires $n$-symbol extension
to achieve $\lo(1/n)$ redundancy even though
the complexity does not scale with $n$ as in Huffman coding.
Schulte and B\"ocherer \cite{homophonic_ff}
proposed
an FF homophonic code which outputs
sequences in a fixed type class, which has redundancy $\lo((\log n)/ n)$.
Honda and Yamamoto \cite{homophonic_isita} proposed a VF homophonic code
with redundancy $\lo(1/n)$ with linear complexity
by combining Shannon-Fano-Elias code and Gray code.
Whereas the code in \cite{homophonic_ff} is easier to handle in real systems
from the nature of FF codes,
the VF code in \cite{homophonic_isita}
is easy to extend to non-i.i.d.~processes.
Thus this VF code can be applied to the coding framework
in \cite{boc_chain}\cite{howto_capacity} even if the channel is not memoryless
where the capacity achieving input distribution is not i.i.d.
A drawback of the scheme in \cite{homophonic_isita} is that
an upper bound on the maximum relative probability gap
\begin{align}
\max_{x,x'\in \cX, x^{n-1}\in\cX^{n-1}}
P_{X_n|X^{n-1}}(x|x^{n-1})/P_{X_n|X^{n-1}}(x'|x^{n-1})\n
\end{align}
has to be known beforehand.
This value is easy to compute for Markov processes of (not very large) order $k$ 
but hard to compute
for general block codes,
which makes application to block codes with structures difficult.

In this paper we propose a new VF homophonic coding scheme,
which encodes a variable-length uniform sequence into
an $n$-bit sequence with some target probability distribution $P_{X^n}$.
Here $X^n$ does not have to be i.i.d.~and
the only requirement is that, as in arithmetic coding,
the conditional probability $P_{X_k|X^{k-1}}(x_k|x^{k-1})$ for each $k \in \mathbb{N}$
be computable for a given $x^k$.
The cost for this advantage is that
the scheme requires one code-symbol decoding delay, but
a decoding error propagates to at most one block with very high probability.
We prove the asymptotic optimality of the scheme
under some conditions
and confirm its performance by simulations for
asymmetric channel coding application.

\section{Preliminaries}\label{sec_pre}
We consider a VF homophonic coding problem
to encode a uniform input sequence
$U^{\infty}=(U_1,U_2,\cdots)\in\{0,1\}^{\infty}$ into a sequence
$x^{\infty}\in\mathcal{X}^{\infty}$ where
$\mathcal{X}=\{0,1,\cdots,|\cX|-1\}$
with target probability distribution $P_{X^{\infty}}$.
A random variable with distribution $P_{X^{\infty}}$ is denoted by
$X^{\infty}=(X_1,X_2,\cdots)$.
In a VF homophonic coding scheme,
variable-length subsequences $u_{1}^{l_1},u_{l_1+1}^{l_2},\cdots$ are
invertibly encoded into fixed-length sequences $x_1^n,x_{n+1}^{2n},\cdots$,
where a subsequence is denoted by, e.g., $x_i^j=(x_i,x_{i+1},\cdots,x_j)$.
We consider the case where blocks $X_1^n,\,X_{n+1}^{2n},\cdots$
are i.i.d., whereas symbols $X_1,\,X_2,\cdots,\,X_n$
in $X_1^n$ may be non-i.i.d.
For this reason, we often write $(X_{1,(1)}^n,\allowbreak X_{1,(2)}^n,\cdots)$
instead of $(X_{1}^n,\,X_{n+1}^{2n},\cdots)$.
We discuss the extension to a general sequence
$X^{\infty}=(X_1,X_2,\cdots)$ in Remark \ref{remark_noniid}.

Let $\phi$ be a (possibly random) encoding function
of a homophonic code.
This code is called
{\it perfect} if
the generated sequence
$\tilde{X}^{\infty}=\phi(U_1^{l_1})\allowbreak
\phi(U_{l_1+1}^{l_1+l_2})\cdots$ exactly follows the distribution $P_{X^{\infty}}$ for
the input sequence $U_1^{l_1}U_{l_1+1}^{l_1+l_2}\cdots$ i.i.d.~from $P_U$.
We measure the gap between the generated sequence and $X^{\infty}$ by the max-divergence
per block, denoted by
\begin{align}
D_{\phi}=
\sup_{m\ge 1}
\frac1m\max_{x^{mn} \in \cX^{mn}}
\log\frac{P_{\tilde{X}^{mn}}(x^{mn})}{P_{X^{mn}}(x^{mn})}\ge 0\per\label{def_gap}
\end{align}
A homophonic code $\phi$ is perfect if and only if
$D_{\phi}=0$.
Note that a codeword $x^n$
with decoding error probablity $p_{\mathrm{e}}(x^n)$
is generated with probability at most $2^{D_{\phi}}P_{X^n}(x^n)$
under the homophonic code $\phi$.
As a result, if the block decoding error probability of some block code is
$P_{\mathrm{e}}$ under the ideal distribution $P_{X^n}$,
the error probability generated by the homophonic code
is bounded by $2^{D_{\phi}}P_{\mathrm{e}}$.

Let $\lfloor r \rfloor_{l}$
denote the first $l\in\mathbb{N}$ bits of the binary expansion of $r\in[0,1)$.
We sometimes identify $\fl{r}_l\in\allowbreak\{0,1\}^l$
with the real number $2^{-l}\lfloor 2^l r \rfloor\in [0,1)$.
The real number corresponding to the $l+1,\,l+2,\cdots$-th bits
is denoted by $\inner{r}_l=\allowbreak 2^l r-\fl{2^lr} \in [0,1)$.
Thus we have $r=\fl{r}_l+2^{-l}\inner{r}_l$.
We also define $\inner{r}=r-\fl{r} \in [0,1)$ for any $r\in \mathbb{R}$.
We
define the cumulative distribution function
given $X_1$
and its inverse
by
\begin{align}
F_{x_1}(x_2^n)&=\sum_{a_2^n \preceq x_2^n}P_{X_2^n|X_1}(a_2^n|x_1)\com\nn
F_{x_1}^{-1}(r)&=\min\{x_2^n: F_{x_1}(x_2^n)> r\}\com\n
\end{align}
where $\preceq$ denotes the lexicographic order and $\min$ is taken under this order.
We write $x_2^n-1$ for the last sequence before $x_2^n$, that is,
the largest sequence $y_2^n$ such that
$y_2^n\precneqq x_2^n$.

\section{Delayed VF Homophonic Code}\label{sec_delay}
In this section we propose a Delayed VF Homophonic (DVFH) code,
which is based on an interval partitioning
as in the cases of arithmetic coding
and the homophonic code in \cite{homophonic_interval}.

First we give an intuition for the DVFH code
before its description.
In a DVFH code, we do not assign
the full information specifying the (variable-length)
input sequence $u^l$ to $x^n$,
and instead,
we assign the information that is required to specify $u^l$
if $(\log|\cX|)$-bits of information
were obtained in addition to $x^l$.
This additional information is assigned to
the first element $x_1$ of the next encoding block,
which causes the code to incur a one-symbol decoding delay.
By this assignment of $(\log|\cX|)$-bits of information,
the distribution of $x_1$ becomes different from $P_{X_1}$
but the remaining sequence $x_2^n$ almost follows $P_{X_2^n|X_1}$.

Let $l(I)\in \mathbb{N}$ for an interval
$I=[\uI,\oI)\subset [0,1)$
denote the largest integer $l$
such that
\begin{align}
\oI \le \fl{\uI}_l+|\mathcal{X}|\cdot 2^{-l}\per\label{cond_li}
\end{align}
In this way the lower and upper bounds of an interval $I$
are always denoted by $\uI$ and $\oI$, respectively.

As detailed later,
the encoder of a DVFH code sends $\tx_2^n$ such that
$F_{\tx_1}(\tx_2^n-1) \le 0.u_1u_2\cdots < F_{\tx_1}(\tx_2^n)$.
Therefore
$\fl{F_{\tx_1}(\tx_2^n-1)}_{l_0}<0.u_1u_2\cdots
\le \fl{F_{\tx_1}(\tx_2^n-1)}_{l_0}+|\mathcal{X}|\cdot 2^{-l_0}$
holds
for $l_0=l([F_{x_1}(x_2^n-1),F_{x_1}(x_2^n)))$.
This implies that 
if the decoder gets additional $(\log |\mathcal{X}|)$-bits of information
then $u_1,u_2,\cdots,u_{l_0}$ can be recovered.

Let $\sigma(\cdot; \mathcal{V})$
be a permutation of $\cX=\{0,1,\cdots,|\mathcal{X}|-1\}$
representing the descending order of $\mathcal{V}=\{v_i\}\in \mathbb{R}^{|\mathcal{X}|}$,
that is,
we have
$v_{\sigma(0;\mathcal{V})}\ge v_{\sigma(1;\mathcal{V})}\ge\cdots
\ge\allowbreak v_{\sigma(|\mathcal{X}|-1;\mathcal{V})}$.
Here ties are broken arbitrarily in a unified manner between
the encoder and decoder.
The encoding and decoding algorithms of
a DVFH code are given in
Algorithms \ref{enc_dvfh} and \ref{dec_dvfh},
where the input $u^{\infty}$ is decoded into $\hat{u}^{\infty}$.
Here the length of an interval $I=[\uI,\,\oI)$ is expressed by
$|I|=\oI-\uI$.
Step \ref{dec_exp} of the decoding algorithm is exception handling
and the condition is never satisfied if the decoder receives the generated sequence $\tilde{x}^n$ 
without error.


First we
give a theorem on the unique decodability
and the error propagation probability of a DVFH code.
\begin{theorem}\label{thm_decodable}
(i) A DVFH code satisfies
$u_{t_{j-1}+1}^{t_{j}}=\hat{u}_{\hat{t}_{j-1}+1}^{\hat{t}_{j}}$
for all $j=1,2,\cdots$.
(ii) Fix $j_0\in\mathbb{N}$ and $y^n\in\mathcal{X}^n$ arbitrarily.
If the
sequence of codewords $\tilde{x}_{(1)}^n,\,\tilde{x}_{(2)}^n,\,\cdots$
is generated
by the encoder
from
uniformly distributed $u^{\infty}\in \{0,1\}^{\infty}$
and a sequence
$\tilde{x}_{(1)}^n,\,\cdots,\tilde{x}_{(j_0-1)}^n,\,
y^n,\,\tilde{x}_{(j_0+1)}^n,\,\tilde{x}_{(j_0+1)}^n,\cdots$
is parsed by the decoder, then
\begin{align}
\Pr[u_{t_{j_0+k}}^{\infty}\neq
\hat{u}_{\hat{t}_{j_0+k}}^{\infty}
]
\le
(4p_{\max})^{k},\quad \forall k\in \mathbb{N}
\label{prob_prop}
\end{align}
where
$p_{\max}=\max_{x^n\in\mathcal{X}^n}P_{X_2^n|X_1}(x_2^n|x_1)$.
(iii) The max-divergence between the target distribution
and the distribution of the generated sequence satisfies
\begin{align}
D_{\phi}
&\le
\log \max_{i \in \{0,1,\cdots,|\cX|-1\}}
\frac{2}{(i+1)\max_{x\in \cX}^{(i)}P_{X_1}(x)}
\label{bound_dist}
\end{align}
where $\max_{x\in \cX}^{(i)}f(x)=f(\sigma(i;\{f(x)\}_{x\in \cX}))$ is the
$i$-th largest value in $\{f(x)\}_{x\in \cX}$.
\end{theorem}
The first part of this theorem shows that this code is uniquely decodable.
The second part shows that
if an error occurred in the $k$-th codeword block $\tx_{1,(k)}^n$
then the error may occur in the decoding of $\tx_{1,(k-1)}^n$ and $\tx_{1,(k)}^n$,
but propagates to that of $x_{1,(k+1)}^n,\,x_{1,(k+2)}^n,\cdots$
with exponentially small probability if $p_{\max}<1/4$.
Here note that
$p_{\max}=\max_{x^n\in\mathcal{X}^n}P_{X_2^n|X_1}(x_2^n|x_1)$
itself is also a value exponentially small in $n$
for usual distributions $P_{X^n}$.
The last part shows that the max entropy is bounded by a constant
independent of $n$.
This means that when a DVFH code is applied to
the generation of biased codewords,
the block decoding error probability is
as most a constant times the error probability under the ideal codeword distribution
as explained in the discussion around \eqref{def_gap}.

\begin{algorithm}[t]
\DontPrintSemicolon
\SetKwInput{Input}{Input}
\SetKwInput{Parameter}{Parameter}
\Input{Uniform message $u^{\infty}\in \{0,1\}^{\infty}$.}
$t_{(1)}:=1,\,k:=0,\,I:=[0,1)$.\;
\For{$j=1,2,\cdots$}{
 $\tx_1:=\sigma(k;\{P_X(i)\}_i)$.\; \label{enc_x}
 $r:=0.u_{t_{(j)}}u_{t_{(j)}+1}\cdots$ and $\tx_2^n:=F^{-1}_{\tx_1}(r)$.\;\label{runif}
 $I':=I\cap [F_{\tx_1}(\tx_2^n-1),\,F_{\tx_1}(\tx_2^n))$.\;\label{enc_same1}
 $l_0:=l(I')$.\;
 $I'_i:=I'\cap [\fl{\uI'}_{l_0}+i2^{-l_0},\,\fl{\uI'}_{l_0}+(i+1)2^{-l_0})$ for
$i\in\cX$.\;
\label{enc_same2}
 Set
$i^* \in \cX$
as the index s.t. $r \in I'_{i^*}$.\;
 $l_1:=\lfloor -\log |I'_{i^*}|\rfloor$.\;
 $I:=[\inner{\uI_{i^*}'}_{l_1},\inner{\oI_{i^*}'}_{l_1})$.\;\label{I_kousin}
 $k:=\sigma^{-1}(i^*; \{I'_{i}\}_i)$.\;\label{enc_k}
 $t_{(j+1)}:=t_{(j)}+l_1$.\;
 Output
$\tx_{(j)}^n:=\tx^n$.\;
}
\caption{Encoding of a DVFH Code}\label{enc_dvfh}
\end{algorithm}%
\begin{algorithm}[t]
\DontPrintSemicolon
\SetKwInput{Input}{Input}
\SetKwInput{Parameter}{Parameter}
\Input{Received sequences $\tx_{(1)}^n,\,\tx_{(2)}^n,\cdots\in \{0,1\}^n$.}
\For{$j=1,2,\cdots$}{
 \If{$j\ge 2$}{
  $k:=\sigma^{-1}(\tx_{1,(j)};\{P_X(i)\}_i)$.\;\label{dec_k}
  $i^*:=\sigma(k,\{I_i'\}_i)$.\;\label{dec_x}
  $l_1:=\lfloor -\log |I'_{i^*}|\rfloor$.\;
  $I:=[\inner{\uI_{i^*}'}_{l_1},\inner{\oI_{i^*}'}_{l_1})$.\;\label{dec_I}
  Output $\hat{u}_{\hat{t}_{(j-1)}}^{\hat{t}_{(j-1)}+l_1-1}:=\fl{\uI'_{i^*}}_{l_1}$.\;\label{out_dec}
  $\hat{t}_{(j)}:=\hat{t}_{(j-1)}+l_1$.\;
 }
 $\tx^n:=\tx_{(j)}^n$.\;
 $I':=I\cap [F_{\tx_1}(\tx_2^n-1),\,F_{\tx_1}(\tx_2^n))$.\;\label{dec_same1}
 {\bf if $I'=\emptyset$ then $I':=[F_{\tx_1}(\tx_2^n-1),\,F_{\tx_1}(\tx_2^n))$}.\;\label{dec_exp}
 $l_0:=l(I')$.\;
 $I'_i:=I'\cap [\fl{\uI'}_{l_0}+i2^{-l_0},\,\fl{\uI'}_{l_0}+(i+1)2^{-l_0})$ for
$i\in\cX$.\;\label{dec_same2}
}
\caption{Decoding of a DVFH Code}\label{dec_dvfh}
\end{algorithm}

Unfortunately, we do not have a theoretical guarantee on the
average input length of a DVFH code.
This is because $r$ at each iteration is uniformly distributed
over an interval $I\subset [0,1)$ rather than $[0,1)$.
For example,
if $I=[0,b)$ for some $b>0$ then
the distribution of $\tilde{x}_2^n$ becomes
{\allowdisplaybreaks[0]%
\begin{align}
\lefteqn{
P_{\tilde{X}_2^n|\tilde{X}_1}(\tilde{x}_2^n|\tx_1)
}\nn
&=
\begin{cases}
\frac{1}{b}P_{\tilde{X}_2^n|\tilde{X}_1}(\tilde{x}_2^n|\tx_1),&
\tx_2^n \precneqq F_{\tx_1}^{-1}(\tx_2^n),
\\
\frac{b-F_{\tx_1}(\tx_2^n-1)}{b},&
\tx_2^n = F_{\tx_1}^{-1}(\tx_2^n),
\\
0,&\mbox{otherwise.}
\end{cases}\n
\end{align}}%
This problem can be avoided by, for example, adding a shared
common random number $r_{(j)}'\in[0,1)$ to $r$ at each iteration,
which makes $r$ uniformly distributed over $[0,1)$.
The problem can also be avoided by
replacing the cumulative conditional distribution $F_{\tx_1}(\tx_2^n)$
with
$\uI+(\oI-\uI)F_{x_1}(\tx_2^n)$,
which is the technique used in \cite{homophonic_isita}.
However, this modification causes error propagation despite the use of a VF code
since $I$ for the current loop depends on all the sequences sent
in the previous loops.

In this paper we propose a modification of the code
to guarantee
a lower bound on
the average input length
without introducing common randomness and
causing error propagation.
This uses
a ``shifted'' cumulative distribution
given $x_1$
where an appropriately fixed value
$s_{x_1}\in [0,1)$ is added
under the mod 1 operation
(see
\eqref{def_smin} and
\eqref{def_sx} in the proof of Theorem \ref{thm_performance}
for the specific value of $s_{x_1}$),
which is expressed as
\begin{align}
F_{x_1}(x_2^n; s_{x_1})
&=
\inner{F_{x_1}(x_2^n)-s_{x_1}}\com\nn
F_{x_1}^{-1}(r;s_{x_1})
&=F_{x_1}^{-1}(\inner{r+s_{x_1}})\per
\n
\end{align}

\begin{theorem}\label{thm_performance}
Consider the homophonic code where
$F_{\tx_1}(\cdot)$ and $F_{\tx_1}^{-1}(\cdot)$
are replaced with 
$F_{\tx_1}(\cdot\,; s_{\tx_1})$ and
$F_{\tx_1}^{-1}(\cdot\, ;s_{\tx_1})$, respectively,
Step \ref{enc_x} in the encoding is replaced with
$\tilde{x}_1:=i^*$ (with an arbitrary initial value $i^*\in\mathcal{X}$),
and
Step \ref{dec_x} in the decoding is replaced with
$\tilde{x}_1:=\tx_{1,(j)}$.
Then (i) and (ii) of Theorem \ref{thm_decodable} still hold
and the max divergence satisfies
\begin{align}
D_{\phi}
&\le
\log \left\{
\frac{2}{\min_{x\in \cX}P_{X_1}(x)}
\right\}\per\n
\end{align}
Furthermore, there exists $\{s_{x_1}\}_{x_1}\in [0,1)^{|\mathcal{X}|}$
such that
the average input length $l_1=t_{(j+1)}-t_{(j)}$
for any $j\in\mathbb{N}$
satisfies
\begin{align}
\E[l_1]
>
\min_{x\in\mathcal{X}}H(X_2^n|X_1=x)-\log((|\mathcal{X}|-1)/2)
\per
\label{bound_rate}
\end{align}
\end{theorem}
The value $s_{x_1}$ that assures \eqref{bound_rate}
is difficult to compute in practice
(although we only have to compute this value once as preprocessing). 
Nevertheless,
this theorem may become one reason
for the near optimal empirical performance
of the DVFH code (without modification) shown in the next section.

As shown in the above theorem,
the average input length of the modified code
is bounded in terms of
$\min_{x\in\mathcal{X}}\allowbreak H(X_2^n|X_1=x)$.
It must be noted that
this value
is much smaller than $H(X^n)$ or
$H(X_2^n|X_1)$
for some $P_{X^n}$,
although 
$X_2^n$ does not heavily depend on $X_1$
in most ``good'' block codes.
It is an important future work to devise a VF homophonic code
which provably achieves the average input lengths $H(X^n)-\lo(1)$
for general distributions $P_{X^n}$.

\begin{remark}{\rm
In a DVFH code
$(\log|\cX|)$-bits of additional information is assigned to $X_{1,(j-1)}$.
We can also consider a code
such that
the additional information
is assigned to $X_{n,(j)}$
by switching the role of $X_{1,(j+1)}$ to $X_{n,(j)}$.
Such a code does not suffer a decoding delay, but
theoretical guarantees
hold in terms of
$\inf_{x_1^{n-1}\in\cX^{n-1}}P_{X_n|X_1^{n-1}}(x|x_1^{n-1})$
instead of $P_{X_1}(x)$  in, e.g., \eqref{bound_dist}.
Thus it is inappropriate to use this code
to generate sequences which heavily depend on past sequences.
However it can be a promising candidate
when $X_1^n$ is i.i.d.~or a Markov process of order $k\ll n$,
which is the case when a homophonic code is applied
to the framework of channel coding in \cite{howto_capacity}.
}\end{remark}

\begin{remark}\label{remark_noniid}{\rm
In the case where
$X_{1,(1)}^n,X_{1,(2)}^n,\cdots$ are not i.i.d., we can obtain
a similar theoretical guarantee by replacing
$P_{X_2^n|X_1}=P_{X_{2,(j)}^n|X_{1,(j)}}$ in the algorithms with
$P_{X_{2,(j)}^n|X_{1,(j)},\{X_{1,(j')}^n\}_{j'=1}^{j-1}}$,
if
\begin{align}
\lefteqn{
\!\!
\inf_{j\in\mathbb{N}}\min_{
 (x_{1,(j)},\,\{x_{1,(j')}^n\}_{j'=1}^{j-1})\in\cX^{1+n(j-1)}
}}\nn
&\qquad\qquad\qquad
 P_{X_{1,(j)}^n|
 \{X_{1,(j')}^n\}_{j'=1}^{j-1}
}
(x_{1,(j)}|\{x_{1,(j')}^n\}_{j'=1}^{j-1})>0\per\n
\end{align}
However, error propagation inevitably occurs
in the coding framework using probabilities
depending on all the past generated blocks.
Thus it is realistic to
set the target probability distribution
to be independent between blocks.
}\end{remark}

\section{Numerical Results}
In this section we first
compare a DVFH code
with other homophonic codes
for i.i.d.~target distributions.
We next apply the DVFH code to polar coding for asymmetric channels.
We used a DVFH code without introducing the modification
considered in Theorem \ref{thm_performance},
which means that
the theoretical guarantee on the coding rate in \eqref{bound_rate} does not hold.

\begin{figure}[t]
  \begin{center}
   \includegraphics[bb=140 205 498 620,clip,angle=270,width=60mm]{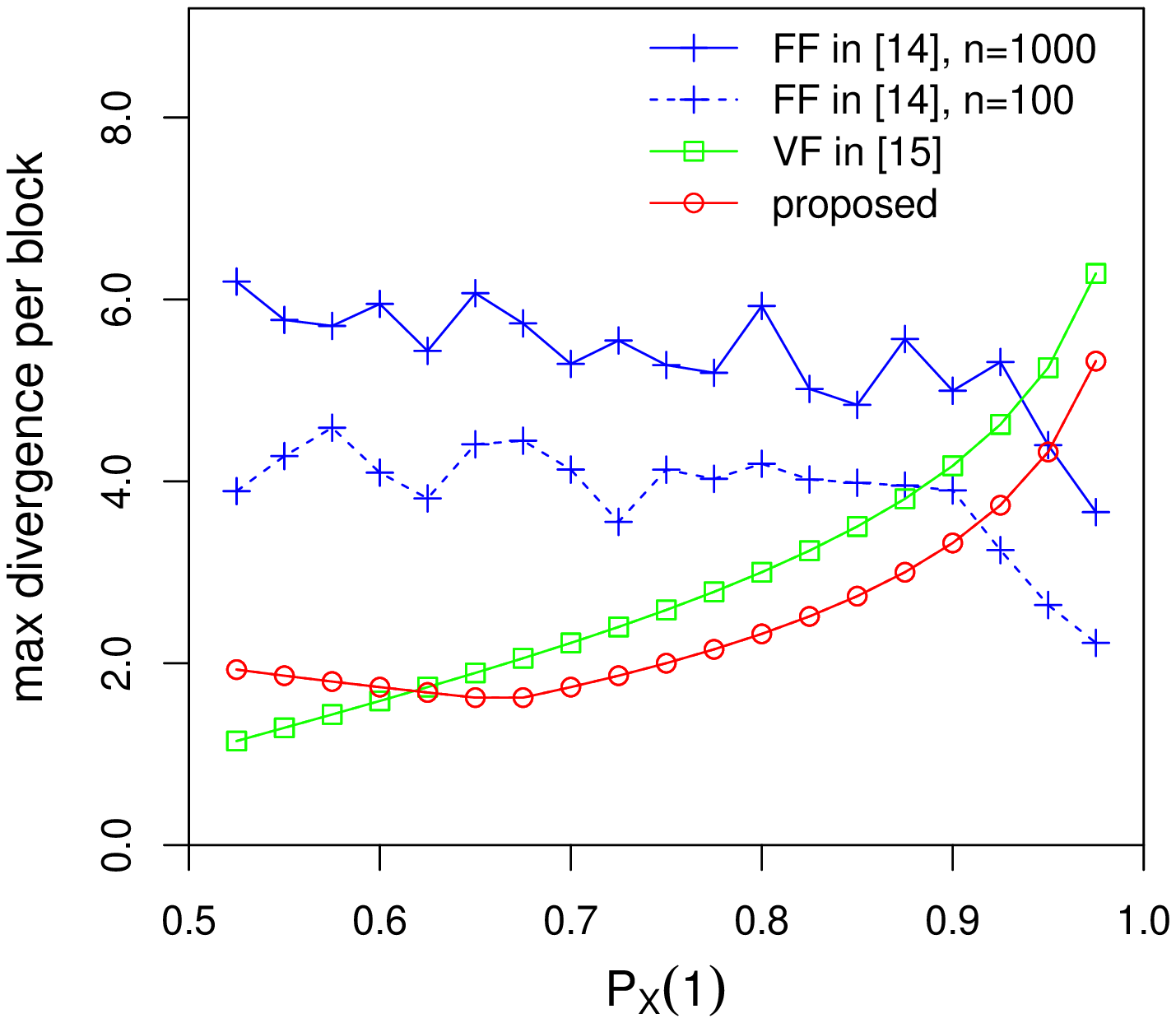}
  \end{center}\vspace{-2.5mm}
  \caption{Max-divergence of homophonic codes for i.i.d.~binary sequences.}\vspace{-3mm}
  \label{fig1}
  \begin{center}
   \includegraphics[bb=115 198 498 620,clip,angle=270,width=60mm]{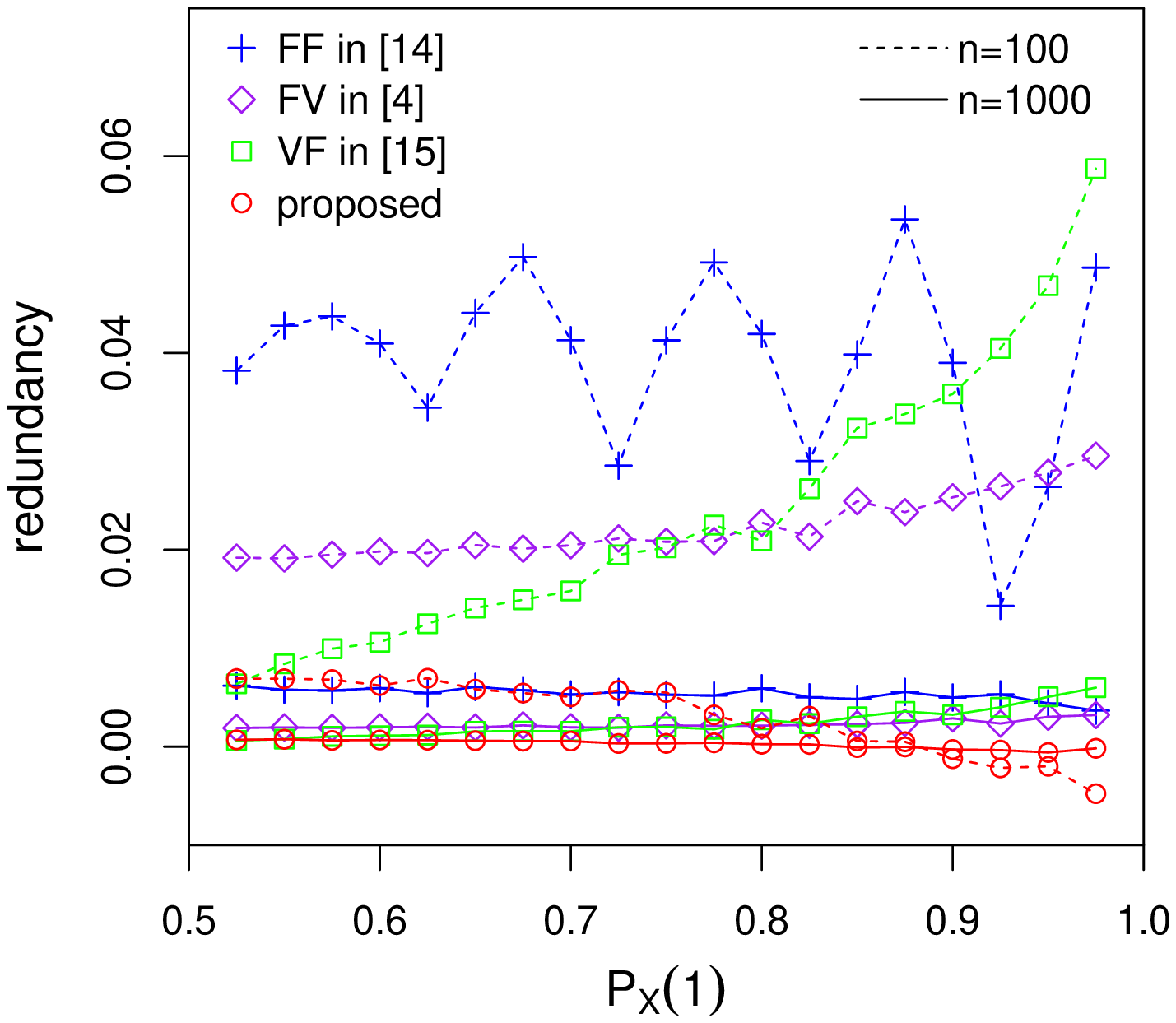}
  \end{center}\vspace{-2.5mm}
  \caption{Redundancy of homophonic codes for i.i.d.~binary sequences.}
  \label{fig2}
\end{figure}%
Fig.~\ref{fig1} is the
upper bound\footnote{We considered the theoretical upper bound instead of
the empirical results for this point since it requires prohibitively large number of samples
to estimate a distribution over $\mathcal{X}^n$.}
 on max-divergence $D_{\phi}$ between
the generated sequences and the target probability.
The plots\footnote{%
The plots of the FF code in this figure
is slightly different from the ISIT version that contained a bug.}
of the FF code in \cite{homophonic_ff} are the exact ones rather than upper bounds.
The sequences generated by the FV code in \cite{homophonic_interval} exactly
follow the target distribution and its plot is not shown in the figure.
Max divergences per block
of the code in \cite{homophonic_isita} and the proposed one are bounded independent of $n$
whereas that of the FF code in \cite{homophonic_ff} is
$\Theta(\log n)$.
Fig.~\ref{fig2} shows the redundancy of the average coding rate which is
$\E[L_{\mathrm{in}}]/n-H(X)$ for VF and FF codes
where $L_{\mathrm{in}}$ is the input length.
For the FV code in \cite{homophonic_interval},
the redundancy is given by
$m/\E[L_{\mathrm{out}}]-H(X)$
where $m$ and $L_{\mathrm{out}}$ are the input and output length, respectively,
and we set $m=\lfloor nH(X)\rfloor$ so that the output length
becomes roughly the same as the FF and VF codes.
Each plot is the average over 10,000 sequential encoding.

As we can see from these figures
the DVFH code achieves a comparable divergence
for most target distributions,
whereas the redundancy is almost zero or below even for
shorter block length.
Here the negative redundancy of the DVFH code
does not contradict the Shannon bound
since the distribution of the generated sequence is slightly different from
the target distribution.

\begin{figure}[t]
  \begin{center}
   \includegraphics[bb=140 205 498 620,clip,angle=270,width=60mm]{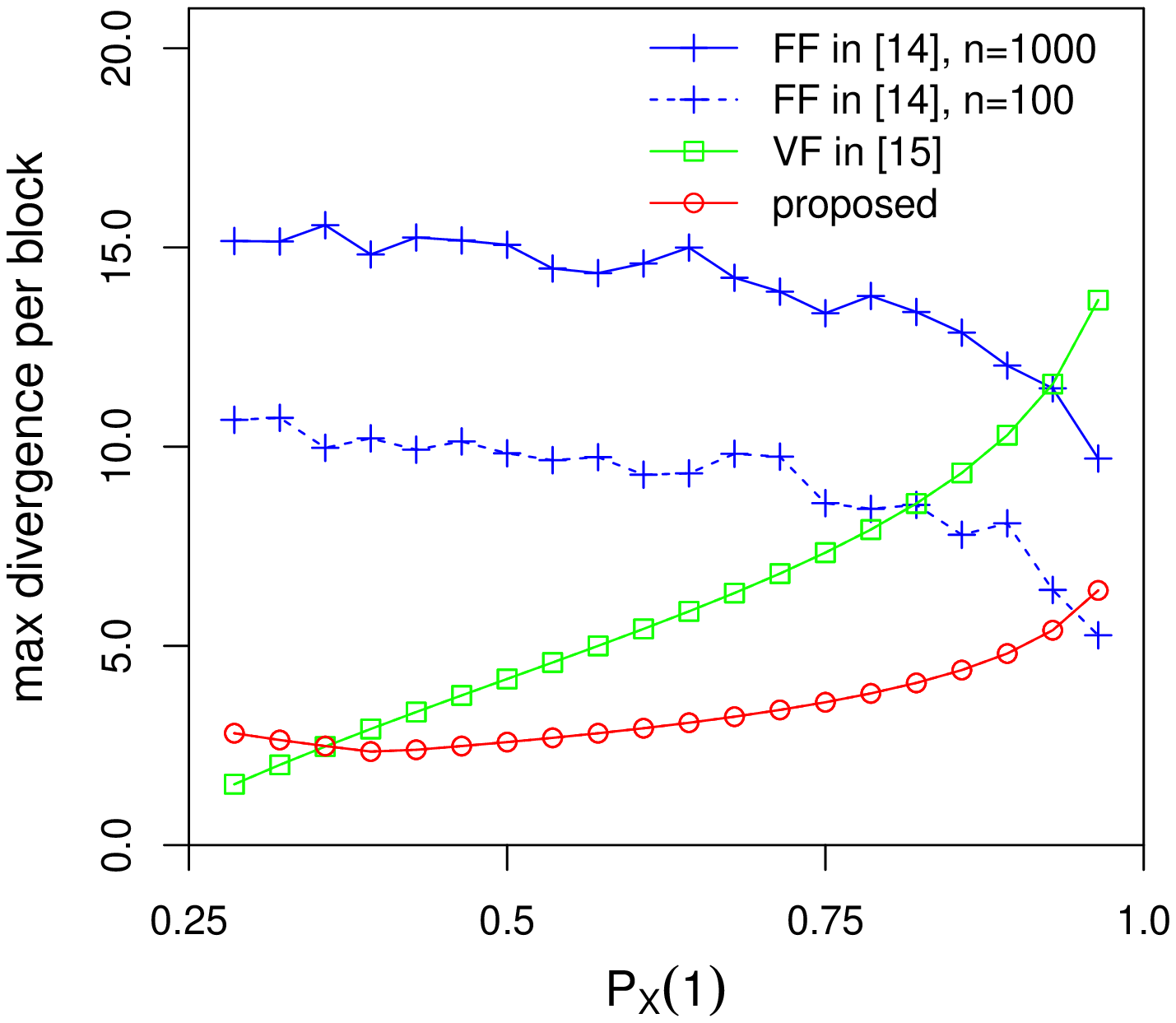}
  \end{center}\vspace{-2.5mm}
  \caption{Max-divergence of homophonic codes for i.i.d.~quaternary sequences.%
}\vspace{-3mm}
  \label{fig41}
  \begin{center}
   \includegraphics[bb=115 198 498 620,clip,angle=270,width=60mm]{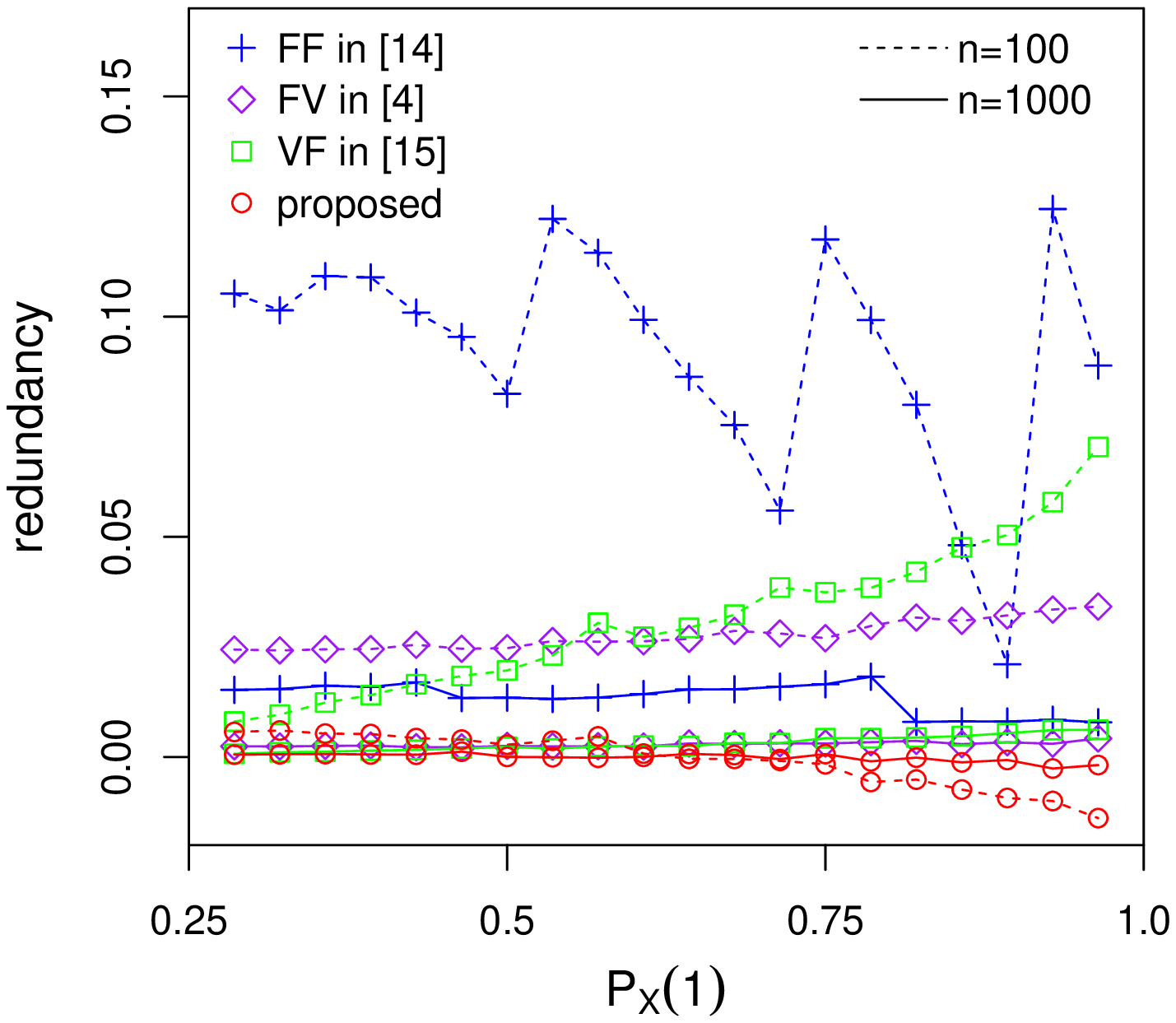}
  \end{center}\vspace{-2.5mm}
  \caption{Redundancy of homophonic codes for i.i.d.~quaternary sequences.%
}
  \label{fig42}
\end{figure}%
Next, Figs.~\ref{fig41} and \ref{fig42} show the
upper bound
on max-divergence $D_{\phi}$
and the redundancy of the average coding rate
for i.i.d.~quaternary sequences with
$P_X(2)=P_X(3)=P_X(4)=(1-P_X(1))/3$, respectively.
A tendency similar to the binary case can be seen from these figures.

Finally we consider an application of homophonic coding to polar codes for asymmetric channels,
where capacity-achieving polar coding schemes are proposed in \cite{polar_honda_trans}
and \cite{wang_memory}.
Whereas the former one
is an FF code using a polar code for lossless compression to realize the optimal input distribution,
the latter one is an FV code based on the FV homophonic code in \cite{homophonic_interval}
which is practically unrealistic because of the problem of error propagation.
Based on this observation we compared the FF length code in \cite{polar_honda_trans}
with the scheme in \cite{wang_memory}
by replacing the FV homophonic code with the DVFH code.
Note that the FF homophonic code in \cite{homophonic_ff} and
the VF homophonic code in \cite{homophonic_isita}
are hard to apply since the target distribution
is not i.i.d.

We considered AWGN channels with 4ASK modulation where input points
are given by $X\in \{-3a,-a,+a,+3a\}$ for $a>0$.
The Signal-to-Noise Ratio (SNR) is set to 10db.
The optimal input distribution is
$(P_X(\pm a),P_X(\pm 3a))=(0.33,\,0.17)$.
The mutual information $I(X;Y)$ between the input and output is
1.582 and 1.628 for the uniform and the above input distribution, respectively.
Note that it is also possible
to optimize the input points as well as the input distribution but
we only considered optimization of the latter one for practicality.
We used polar codes
over $\mathcal{X}=\mathrm{GF}(4)$,
with block lengths $2^9,\,2^{11}$ and $2^{13}$.

\begin{figure}[t]
  \begin{center}
   \includegraphics[bb=130 205 500 615,clip,angle=270,width=60mm]{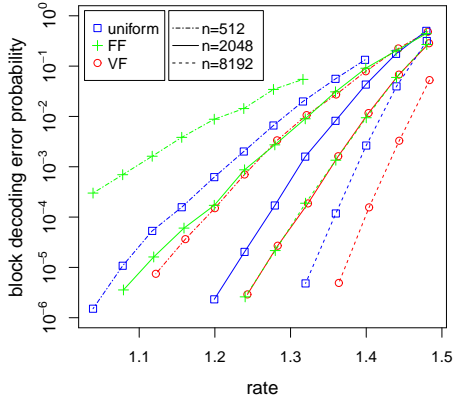}
  \end{center}\vspace{-2.5mm}
  \caption{Block decoding error probability of the original polar code with uniform input,
the capacity achieving FF one in \cite{polar_honda_trans}
and the one in \cite{wang_memory} where the FV homophonic code in \cite{homophonic_interval}
is replaced with the proposed VF code.}
  \label{fig3}
\end{figure}%
Fig.~\ref{fig3} shows the decoding error probabilities of the above two polar coding schemes and the
original polar code with the uniform input distribution.
Here note that
there is some arbitrariness on the performance measure
of the scheme using a DVFH code.
First, decoding error propagates to roughly $3/4$ blocks
(in the 4-ary case), that is, one decoding error of this code
corresponds roughly to $7/4$-block errors of other codes.
Second, the input length for each codeword is variable
and there exists a correlation between the message length
and the error probability of the codeword.
Based on this observation,
we plotted twice the empirical decoding error probability to conservatively
evaluate the scheme using a DVFH code.

As we can see from the figure
the performance of the FF code in \cite{polar_honda_trans} 
is much worse than the performance of
the uniform input polar code
for moderate block lengths,
although it is theoretically assured to be better than the uniform input
asymptotically.
On the other hand, the VF polar code using the DVFH code
significantly outperforms the polar code with the uniform input.

\section{Proof of Theorem \ref{thm_performance}}\label{sec_proof}
In this section we prove
Theorems \ref{thm_decodable}
and \ref{thm_performance}.
We prove these theorems based on the following lemma.
\begin{lemma}\label{lem_unif}
At Step \ref{runif} of the encoding,
$|I|> 1/2$ holds and
$r$ is uniformly distributed over $I$
given $\tx_{(1)}^n,\,\tx_{(2)}^n,\cdots,\tx_{(j-1)}^n$.
\end{lemma}
\begin{proof}
$|I|>1/2$ straightforwardly follows from
Steps \ref{enc_same2}--\ref{I_kousin}.
The latter part is proved by induction.
For $j=1$ this proposition holds from $I=[0,1)$ and the uniformity of
$(u_1,u_2,\cdots)$.
Next, assume that the proposition holds for $j\le j_0-1$.
Then, given $\tx_{(j_0)}^n=\tx^n$,
$r$ is uniformly distributed over
$I'$ in Step \ref{enc_same1}.
Thus, $0.u_{t_{(j)}+l_1}u_{t_{(j)}+l_1+1}\cdots$
is uniformly distributed over $I:=[\inner{\uI_{i^*}'}_{l_1},\inner{\oI_{i^*}'}_{l_1})$,
which implies that the proposition holds for $j=j_0$.
\end{proof}

\begin{proof2}{of Theorem \ref{thm_decodable}}
(i)
Construction of $I',l_0,\,\{I_i'\},\,l_1$ and $I$ is the
same between the encoding and the decoding, provided that
$i^*$ is the same between them.
In addition,
from relation between
Steps \ref{enc_k} and \ref{enc_x} in the encoding and
Steps \ref{dec_k} and \ref{dec_x} in the decoding,
$i^*$ for each $j=j_0$ in the encoding is correctly recovered in the decoding
at Step \ref{dec_x} for $j=j_0+1$.
Since $u_{t_{(j-1)}}^{t_{(j-1)}+l_1-1}=\fl{\uI'_{i^*}}_{l_1}$ holds from $r\in I_{i^*}'$,
$u_{t_{(j-1)}}^{t_{(j-1)}+l_1-1}$
is correctly recovered by Step \ref{out_dec} in the decoding,

(ii)
A decoding error
$u_{t_{j+k}}^{t_{j+k+1}-1}\neq\allowbreak\hat{u}_{\hat{t}_{j+k}}^{\hat{t}_{j+k+1}-1}$
occurs only if
$I'$ after Step \ref{enc_same1} of the encoding and
$I'$ after Step \ref{dec_exp} of the decoding are different for $j=j_0+k$.
From Lemma \ref{lem_unif}, $I'\neq [F_{\tx_1}(\tx_2^n-1),\,F_{\tx_1}(\tx_2^n))$
occurs in the encoding with probability at most
$\max_{x_2^n}P_{X_2^n|X_1}(x_2^n|\tx_1)/{|I|}
\le
2p_{\max}$
given that $I'$ are different between the encoding and decoding for $j=j_0+k-1$.
The same result holds for the decoding and therefore
$I'$ are different to each other with probability at most $4p_{\max}$,
which proves \eqref{prob_prop}.

(iii)
For each $i_0\in\mathcal{X}$,\,
$\tilde{x}_1=x_{\sigma(i_0;\{P_X(i)\}_i)}$ if and only if
$r\in I_{\sigma(i_0;\{I_i'\}_i)}'$
from the encoding algorithm,
which holds with probability $|I_{\sigma(i_0;\{I_i'\}_i)}'|/|I'|$ from Lemma \ref{lem_unif}.
Since $|I_{\sigma(0;\{I_i'\}_i)}'|\ge |I_{\sigma(1;\{I_i'\}_i)}'| \ge \cdots
\ge |I_{\sigma(|\mathcal{X}|-1;\{I_i'\}_i)}'|$,
we have
\begin{align}
P_{\tilde{X}_{1,(j)}}(\sigma(i_0;\{P_X(i)\}_i))
&=\frac{I_{\sigma(i_0;\{I_i'\}_i)}'}{\sum_{i\in\cX}I_i}
\le
\frac{1}{i_0+1}\per\n
\end{align}
We also have
\begin{align}
P_{\tilde{X}_{2,(j)}^n|\tilde{X}_{1,(j)}}(\tx_2^n|\tx_1)
\le
\frac{P_{X_{2}^n|X_1}(\tx_2^n|\tx_1)}{|I|}
<
2P_{X_2^n|X_1}(\tx_2^n|\tx_1)\n
\end{align}
from Lemma \ref{lem_unif}.
Thus we obtain \eqref{bound_dist} by
\begin{align}
D_{\phi}
&=\sup_{j\ge 1}
\frac1j\max_{x^{jn} \in \mathcal{X}^{jn}}
\log\frac{P_{\tilde{X}^{jn}}(x^{jn})}{P_{X^{jn}}(x^{jn})}\nn
&\le
\log
\max_{i_0\in\mathcal{X}}
\frac{\frac{1}{i0+1}}{P_{X_1}(x_{\sigma(i_0;\{P_X(i)\}_i)})
}
\max_{x_2^n\in\mathcal{X}_2^n}
\frac{2P_{X_2^n|X_1}(x_2^n|x_1)}{P_{X_2^n|X_1}(x_2^n|x_1)}\nn
&=
\log \max_{i \in \{0,1,\cdots,|\mathcal{X}|-1\}}
\frac{2}{(i+1)\max_{x\in \cX}^{(i)}P_{X_1}(x)}
\per\dqed\n
\end{align}
\end{proof2}

\begin{proof}[Proof of Theorem \ref{thm_performance}]
As the former part is almost the same as the proof of
Theorem \ref{thm_decodable}, we only prove \eqref{bound_rate}.
Since $l(I)$ is the largest $l$ satisfying
\eqref{cond_li}, we have
\begin{align}
\oI
&> \fl{\uI}_{l(I)+1}+|\mathcal{X}|\cdot 2^{-(l(I)+1)}\nn
&> \uI-2^{-(l(I)+1)}+|\mathcal{X}|\cdot 2^{-(l(I)+1)}\nn
&= \uI+(|\mathcal{X}|-1)\cdot 2^{-(l(I)+1)}\com\n
\end{align}
which implies
\begin{align}
l(I)>
-\log |I|+\log ((|\mathcal{X}|-1)/2)
\per\n
\end{align}
Therefore,
the length of the assigned message is given by
\begin{align}
l_1
&\ge-\log 2^{-l(I')}\nn
&>-\log|I'|+\log((|\mathcal{X}|-1)/2)\nn
&\ge -\log P_{X_2^n|X_1}(\tx_{2,(j)}^n|\tx_{1,(j)})+\log((|\mathcal{X}|-1)/2)
\per\label{l_lower}
\end{align}

Now we consider the output distribution
$P_{\tilde{X}_2^n|\tilde{X}_1}$.
Recall that $r$ is uniformly distributed over $I$
by Lemma \ref{lem_unif}.
We have $I=[a,1)$ for some $a\in [0,1/2)$ if $\tx_{1,(j)}=0$ holds,
$I=[0,b)$ for some $b\in [1/2,1)$ if $\tx_{1,(j)}=|\cX|-1$ holds
and $I=[0,1)$ otherwise.
Thus, in the case where $\tx_{1,(j)}\notin \{0,|\mathcal{X}|-1\}$ we have
\begin{align}
\E[l_1]
&\ge \E_{X_2^n}[-\log P_{X_2^n|X_1}(X_2^n|\tx_{1,(j)})]+
\log((|\mathcal{X}|-1)/2)\nn
&\ge H(X_2^n|X_1=\tx_{1,(j)})+
\log((|\mathcal{X}|-1)/2)
\n
\end{align}
for any $s_{\tx_{1,(j)}}$
and we consider the other case in the following.

Now consider the case $\tx_{1,(j)}=0$.
Let
\begin{align}
\bar{f}=
\int_0^1 f(r)\rd r\com\quad\;
s_0=\argmax_{s\in[0,1)}\left\{\int_0^s (f(r)-\bar{f})\rd r\right\}
,\label{def_smin}
\end{align}
for $f(r)=-\log P_{X_2^n|X_1}(F_{x_1}^{-1}(r)|x_1)$.
Then we have
\begin{align}
\lefteqn{
\E_{\tilde{X}_2^n}[-\log P_{X_2^n|X_1}(\tilde{X}_2^n|\tx_{1,(j)})]
}\nn
&=
\frac{1}{1-a}\int_{a}^1 (-\log P_{X_2^n|X_1}(F^{-1}_{0}(r;s_0)|\tx_{1,(j)}))\rd r\nn
&=
\bar{f}+\frac{\int_{\inner{a+s_0}}^{s_0}(f(r)-\bar{f})\rd r}{1-a}\label{int_loop}\\
&=
\bar{f}+\frac{\int_{0}^{s_0}(f(r)-\bar{f})\rd r-\int_0^{\inner{a+s_0}}(f(r)-\bar{f})\rd r}{1-a}
\label{use_min}\\
&\ge
\bar{f}
=
H(X_2^n|X_1=\tx_{1,(j)})\com\n
\end{align}
where \eqref{int_loop} and \eqref{use_min}
follow from $\int_0^1 (f(r)-\bar{f})\rd r=0$
and \eqref{def_smin}, respectively.
In the case $\tx_{1,(j)}=|\mathcal{X}|-1$
we obtain
$\E_{\tilde{X}_2^n}[-\log P_{X_2^n|X_1}(\tilde{X}_2^n|\tx_{1,(j)})]\ge
H(X_2^n|X_1=\tx_{1,(j)})$
in the same way by letting
\begin{align}
s_{|\mathcal{X}|-1}=\argmin_{s\in[0,1)}\left\{\int_0^s (f(r)-\bar{f})\rd r\right\}\per
\label{def_sx}
\end{align}
We obtain \eqref{bound_rate} by
combining this result with
\eqref{l_lower}.
\end{proof}

\section{Conclusion}
In this paper we proposed a variable-to-fixed length homophonic code, DVFH code,
which is easily applied to channel coding for asymmetric channels.
This code can decode each block with one code-symbol decoding delay.
The max-divergence of the generated sequence is bounded by a constant.
The average input length of the code is very close
the entropy in the simulation and
it is shown to be asymptotically larger than
the worst-case conditional entropy
under a slight modification of the code.
An important future work is to construct a
code such that input length provably achieves the entropy rather than
the conditional one.

\section*{Acknowledgment}
This work was supported in part by JSPS KAKENHI Grant Number 16H00881,


\end{document}